\documentclass[11pt]{article}

\usepackage[margin=0.8in]{geometry}

\usepackage{hyperref}
\usepackage{appendix}

\usepackage{graphicx,epsfig,amsmath,latexsym,amssymb,verbatim}

\newcommand{\extra}[1]{}

\usepackage{amsmath,amssymb,amsfonts}
\usepackage{amsthm}

\newtheorem{theorem}{Theorem}[section]

\newtheorem{corollary}[theorem]{Corollary}

\newtheorem{lemma}[theorem]{Lemma}

\newtheorem{proposition}[theorem]{Proposition}

\theoremstyle{remark}

\def\squareforqed{\hbox{\rlap{$\sqcap$}$\sqcup$}}
\def\qed{\ifmmode\squareforqed\else{\unskip\nobreak\hfil
\penalty50\hskip1em\null\nobreak\hfil\squareforqed
\parfillskip=0pt\finalhyphendemerits=0\endgraf}\fi}
\def\endenv{\ifmmode\;\else{\unskip\nobreak\hfil
\penalty50\hskip1em\null\nobreak\hfil\;
\parfillskip=0pt\finalhyphendemerits=0\endgraf}\fi}

\renewenvironment{proof}{\noindent \textbf{{Proof~} }}{\qed\medskip}
\newenvironment{proof+}[1]{\noindent \textbf{{Proof #1~} }}{\qed\medskip}

\mathchardef\ordinarycolon\mathcode`\:
\mathcode`\:=\string"8000
\def\vcentcolon{\mathrel{\mathop\ordinarycolon}}
\begingroup \catcode`\:=\active
  \lowercase{\endgroup
  \let :\vcentcolon
  }

\newcommand{\nc}{\newcommand}
\nc{\rnc}{\renewcommand}
\nc{\beq}{\begin{equation}}
\nc{\eeq}{{\end{equation}}}
\nc{\beqa}{\begin{eqnarray}}
\nc{\eeqa}{\end{eqnarray}}
\nc{\lbar}[1]{\overline{#1}}
\nc{\bra}[1]{\langle#1|}
\nc{\ket}[1]{|#1\rangle}
\nc{\ketbra}[2]{|#1\rangle\!\langle#2|}
\nc{\braket}[2]{\langle#1|#2\rangle}
\nc{\proj}[1]{| #1\rangle\!\langle #1 |}
\nc{\avg}[1]{\langle#1\rangle}
\nc{\smfrac}[2]{\mbox{$\frac{#1}{#2}$}}
\nc{\tr}{\operatorname{tr}}
\nc{\tracedist}[1]{\Delta_{}\!\left( #1 \right)}
\nc{\fid}[1]{F\!\left( #1 \right)}

\nc{\ox}{\otimes}
\nc{\dg}{\dagger}
\nc{\dn}{\downarrow}
\nc{\cA}{{\cal A}}
\nc{\cB}{{\cal B}}
\nc{\cC}{{\cal C}}
\nc{\cD}{{\cal D}}
\nc{\cE}{{\mathcal E}}
\nc{\cF}{{\cal F}}
\nc{\cG}{{\cal G}}
\nc{\cH}{{\cal H}}
\nc{\cI}{{\cal I}}
\nc{\cJ}{{\cal J}}
\nc{\cK}{{\cal K}}
\nc{\cL}{{\cal L}}
\nc{\cM}{{\cal M}}
\nc{\cN}{{\cal N}}
\nc{\cO}{{\cal O}}
\nc{\cP}{{\cal P}}
\nc{\cR}{{\cal R}}
\nc{\cS}{{\cal S}}
\nc{\cT}{{\cal T}}
\nc{\cU}{{\cal U}}
\nc{\cX}{{\cal X}}
\nc{\cY}{{\cal Y}}
\nc{\cZ}{{\cal Z}}

\nc{\entI}{{\bf I}}
\nc{\entIarrow}{{\bf I}^{\leftarrow}}
\nc{\entH}{{\bf H}}
\nc{\entS}{{\bf S}}
\nc{\entHmin}{\mathbf{H}_{\min}}
\nc{\entHmax}{\mathbf{H}_{\max}}

\nc{\binent}{h_2}

\nc{\entF}{{\bf E}_f}

\nc{\isom}{\simeq}

\nc{\rank}{\operatorname{rank}}
\nc{\rar}{\rightarrow}
\nc{\lrar}{\longrightarrow}
\nc{\polylog}{\operatorname{polylog}}
\nc{\poly}{\operatorname{poly}}
\nc{\1}{\mathbb{I}}

\nc{\weight}{\textbf{w}}
\nc{\hamdist}{d_{H}}

\def\e{\epsilon}

\nc{\Sp}{{{\mathbb S}}}
\nc{\RR}{{{\mathbb R}}}
\nc{\CC}{{{\mathbb C}}}
\nc{\FF}{{{\mathbb F}}}
\nc{\NN}{{{\mathbb N}}}
\nc{\ZZ}{{{\mathbb Z}}}
\nc{\PP}{{{\mathbb P}}}
\nc{\QQ}{{{\mathbb Q}}}
\nc{\UU}{{{\mathbb U}}}
\nc{\OO}{{{\mathbb O}}}
\nc{\EE}{{{\mathbb E}}}
\nc{\id}{{\operatorname{id}}}

\nc{\qubitchannel}{\id_2}
\nc{\bitchannel}{\overline{\id}_2}

\nc{\be}{\begin{equation}}
\nc{\ee}{{\end{equation}}}
\nc{\bea}{\begin{eqnarray}}
\nc{\eea}{\end{eqnarray}}
\nc{\<}{\langle}
\rnc{\>}{\rangle}
\nc{\Hom}[2]{\mbox{Hom}(\CC^{#1},\CC^{#2})}
\nc{\rU}{\mbox{U}}

\nc{\ob}[1]{#1}

\newcommand{\eqdef}	{\stackrel{\textrm{def}}{=}}

\newcommand{\exc}[2]	{\underset{#1}{\mathbf{E}}\left\{ #2 \right\}}

\newcommand{\prc}[2]	{\underset{#1}{\mathbf{P}}\left\{ #2 \right\}}

\nc{\unif}{\textrm{unif}}

\nc{\inter}{\textrm{int}}
\nc{\ed}{\textrm{ed}}

\nc{\grade}{\mathsf{G}}

\nc{\pguess}{P_{guess}}

\nc{\barA}{\overline{A}}
\nc{\barB}{\overline{B}}
\nc{\barC}{\overline{C}}
\nc{\barD}{\overline{D}}
\nc{\barR}{\overline{R}}
\nc{\barX}{\overline{X}}
\nc{\barY}{\overline{Y}}
\nc{\barU}{\overline{U}}

\nc{\barrho}{\overline{\rho}}

\nc{\barp}{\overline{p}}

\nc{\Pos}{\mathrm{Pos}}

\newcommand{\eps}{\epsilon}

\newcommand{\channel}{\mathrm{S}}

\newcommand{\NS}{\mathrm{NS}}

\newcommand{\capacity}{\mathrm{C}}
\newcommand{\greedy}{\mathrm{greedy}}
\newcommand{\quantum}{\mathrm{Q}}

\newcommand{\dist}{\mathrm{D}}

\usepackage{xcolor}

\begin{document}

\title{Algorithmic Aspects of Optimal Channel Coding}

\author{Siddharth Barman\thanks{Indian Institute of Science. {\tt barman@csa.iisc.ernet.in}} \and Omar Fawzi\thanks{UMR 5668 LIP - ENS Lyon - CNRS - UCBL - INRIA, Universit\'e de Lyon. {\tt omar.fawzi@ens-lyon.fr}}}

\date{}

\maketitle

\begin{abstract}
A central question in information theory is to determine the maximum success probability that can be achieved in sending a fixed number of  messages over a noisy channel. This was first studied in the pioneering work of Shannon who established a simple expression characterizing this quantity in the limit of multiple independent uses of the channel. Here we consider the general setting with only one use of the channel. We observe that the maximum success probability can be expressed as the maximum value of a submodular function. Using this connection, we establish the following results:
\begin{enumerate}
\item There is a simple greedy polynomial-time algorithm that computes a code achieving a $(1-e^{-1})$-approximation of the maximum success probability. The factor $(1-e^{-1})$ can be improved arbitrarily close to $1$ at the cost of slightly reducing the number of messages to be sent. Moreover, it is {\rm NP}-hard to obtain an approximation ratio strictly better than $(1-e^{-1})$ for the problem of computing the maximum success probability.
\item Shared quantum entanglement between the sender and the receiver can increase the success probability by a factor of at most $\frac{1}{1-e^{-1}}$. In addition, this factor is tight if one allows an arbitrary {non-signaling box} between the sender and the receiver.
\item We give tight bounds on the one-shot performance of the {meta-converse of Polyanskiy-Poor-Verd\'{u}}. 
\end{enumerate}
\end{abstract}


\section{Introduction}
One of the central threads of research in information theory is the study of the minimum error probability that can be achieved in sending a fixed number of messages over a noisy channel. This task can be formulated as the maximization---over valid encoders and decoders---of the probability of correctly determining the sent message (which we refer to as success probability for the rest of the paper). In his seminal work, Shannon~\cite{Sha48} showed that for $n$ independent copies of the channel, this question is almost completely answered by a single number, $\capacity$, the capacity of the channel. In fact, for a number of messages $k(n)$ satisfying $\sup_n \frac{\log k(n)}{n} < \capacity$, the maximum success probability tends to $1$ as $n$ tends to infinity, and when $\inf_n \frac{\log k(n)}{n} > \capacity$, the maximum success probability tends to $0$ as $n$ tends to infinity~\cite{Wol57}. 



Here, we study the algorithmic aspects of determining the optimal encoder and decoder which lead to the maximum success probability over a noisy channel in the non-asymptotic regime. Recently, in the information theory literature there has been significant interest in understanding the non-asymptotic behavior when the number of channel uses $n$ is finite~\cite{Str62,Hay09,PPV10,TT13}. But, instead of focusing  on closed-form expressions for the maximum rate at which information can be transmitted, we rather ask how well can the optimal rates of transmission be computed with an efficient algorithm. One way to formulate the computational problem is as follows. The input of the problem is an integer $k$ (which denotes the total number of messages) together with the description of a channel $W$ that maps elements in $X$ to elements in $Y$; specifically, for each $y \in Y$ and $x \in X$, we have $W(y|x)$, which is the probability of receiving the symbol $y$ when the symbol $x$ is transmitted. The objective is to determine the maximum success probability,\footnote{In this paper, we focus on the success probability on average over the $k$ messages.} $\channel(W,k)$, for transmitting $k$ distinct messages using the channel $W$ (only once). This algorithmic formulation leads us to interesting implications that are described below. We refer to this problem of determining the maximum success probability of a given noisy channel as the optimal channel coding problem.




Our more specific findings can be summarized as follows:
\begin{itemize}
\item We observe that the problem of computing the optimal success probability $\channel(W,k)$ corresponds to a submodular maximization problem with cardinality constraints (Proposition~\ref{prop:s_submodular}), which implies that this quantity can be efficiently approximated using a simple greedy algorithm that achieves a $(1-e^{-1})$-approximation ratio. As the maximum coverage problem can be reduced to the optimal channel coding problem, we also find that it is {\rm NP}-hard to approximate $\channel(W,k)$ within a factor larger than $(1-e^{-1})$.
\item The natural linear programming relaxation of the optimal channel coding problem is well-studied in the information theory literature under a different name. It corresponds to the well-known meta-converse of Polyanskiy-Poor-Verd\'{u} (PPV)~\cite{PPV10} which puts an upper bound on the maximum success probability of sending $k$ messages.\footnote{Technically, the PPV bound gives an upper bound on the number of messages for a desired success probability. It is however simple to adapt it to maximizing success probability for a fixed number of messages. The reason for the usefulness of the meta-converse is that it can be analytically evaluated for many settings of interest in particular for $n$ independent channel uses for non-asymptotic values of $n$; see~\cite{Pol10,Tan14} and references therein for an overview of the active area of finite blocklength analysis.} Matthews~\cite{Mat12} showed that this linear programming relaxation corresponds to the maximum success probability, $\channel^{\NS}(W,k)$, when the sender and the receiver are allowed to share any non-signaling box, i.e., a (hypothetical) device taking inputs from both parties and generating outputs for both parties in a manner that makes it, by itself, useless for communication. Our main finding is an upper bound on the integrality gap for the linear programming relaxation of the optimal channel coding problem (Theorem~\ref{thm:ns-advantage}). In particular, for any  channel $W$ and integers $k$ and $\ell$, we have:
\begin{align}
\label{eq:main-ineq}
\frac{k}{\ell}\left(1- e^{-\ell/k} \right) \channel^{\NS}(W, k) \leq \channel(W,\ell) \leq \channel^{\NS}(W,\ell) \ .
\end{align}
When $\ell = k$, this inequality says that the ratio of the optimal success probability and the meta-converse is at least $(1-e^{-1})$. More generally, if a better guarantee on the success probability is desired, this can be achieved at the expense of taking a number of messages $\ell$ that is slightly smaller than $k$. For example, if we take $k = 2^t$ and $\ell = \frac{2^t}{t}$, we obtain after simplification
\begin{align}
\label{eq:main-ineq-ex}
\left(1 - \frac{1}{2 t} \right) \channel^{\NS}(W, 2^{t}) \leq \channel(W, \frac{2^t}{t}) \leq \channel^{\NS}(W,\frac{2^t}{t}) \ .
\end{align}
The bound~\eqref{eq:main-ineq} can be seen as a bicriterion upper bound on the integrality gap, highlighting the tradeoff between the two important parameters: success probability and number of messages. We note that it is important for our applications to analyze the linear programming relaxation and not only understand the performance of the greedy algorithm for the optimal channel coding problem. We give two algorithmic proofs of this result. The first one by analyzing the greedy algorithm which can be done by combining a result of~\cite{CFN77,NWF78} together with an important observation that can be found in~\cite[Theorem 1.5]{KG12}. We provide a self-contained elementary proof of the result. In the second proof, we analyze the coding strategy which is commonly used in achievability bounds: a random code chosen according to a distribution given by the meta-converse. This second analysis is done using standard randomized rounding techniques. Moreover, using a family of examples from~\cite{AS04_pipage} gives a family of channels for which the guarantees in~\eqref{eq:main-ineq} are tight.
\item As quantum entanglement cannot be used for signaling, the inequality~\eqref{eq:main-ineq} puts a limit on the usefulness of entanglement between the sender and the receiver for the problem of coding over a noisy (classical) channel. The fact that entanglement could improve success probabilities was highlighted by~\cite{PLMKR11,CLMW10}. In this paper we in fact obtain an explicit upper bound on the improvement that can be achieved via entanglement.

The bound~\eqref{eq:main-ineq} setting $\ell=k$ addresses a question asked by Hemenway et al.~\cite{HMSW13} who proved that for the special case of transmitting a single bit, i.e., $k=2$, the ratio of the quantum advantage to the classical advantage is at most $2$.\footnote{Note that the case $k=2$, the quantity $\channel(W,2)$ can be written in a very explicit form as a function of the maximum total variation distances between all pairs of distributions $W(.|x_1)$ and $W(.|x_2)$, see~\cite[Proposition 1]{HMSW13} for more details.} For an explicit generalization of their result to arbitrary values of $k$ see inequality~\eqref{eq:centered_ineq}, which is a consequence of our main theorem.
\end{itemize}





\subsection{Outline}
In Section~\ref{sect:submod} we establish that the optimal channel coding problem corresponds to submodular maximization with cardinality constraints. Though, in and of itself, this connection is direct, it provides a useful starting point. In particular, we extend this connection in Section~\ref{sect:ns} to obtain interesting implications for channel coding with non-signaling boxes and in particular quantum entanglement.

\section{Optimal channel coding as a submodular maximization problem}
\label{sect:submod}
Given a noisy channel $W$ whose input and output alphabets are $X$ and $Y$ respectively, along with an integer $k$, our aim is to build an encoder and a decoder that can transmit $k$ distinct messages with the smallest error probability averaged over the messages. We can describe a (possibly randomized) encoder $e$ taking as input message $i \in [k]$ and mapping it to $x \in X$ with probability $e(x|i)$. Similarly, a decoder $d$ takes as input some $y \in Y$ and outputs $i \in [k]$ with probability $d(i|y)$. The maximum success probability $\channel(W,k)$ of sending $k$ messages using the noisy channel $W$ can be written as the following optimization program over encoding and decoding maps.
\begin{align}\label{eq:channel_classical}
\begin{aligned}
\channel(W, k)\eqdef\;& \underset{e, d}{\text{maximize}}
& & \frac{1}{k} \sum_{x,y,i} e(x|i)W(y|x) d(i|y) \\
& \text{subject to}
& &\sum_{x} e(x|i) = 1 \quad \forall i \in [k] \\
& 
& &\sum_{i} d(i|y) = 1 \quad \forall y \in Y \\
& 
& & 0 \leq e(x|i) \leq 1 \quad \forall (x,i) \in X \times [k] \\ 
&
& & 0 \leq d(i|y) \leq 1 \quad \forall (i,y) \in [k] \times Y\,.
\end{aligned}
\end{align}

The next proposition is a simple but important observation that the problem described in \eqref{eq:channel_classical} is of a very well-studied type: maximizing a submodular function subject to a cardinality constraint.

\begin{proposition}
\label{prop:s_submodular}
Let $W$ be a channel, with input alphabet $X$ and output alphabet $Y$, and $k \geq 1$ be an integer. Then we have
\begin{align}
\label{eq:s_submodular}
\channel(W, k) = \frac{1}{k}\max_{S \subseteq X, |S| \leq k} f_W(S) \quad \text{where} \quad f_W : 2^X \to \RR_+ \text{ is defined by } f_W(S) = \sum_{y \in Y} \max_{x \in S} W(y|x) \ .
\end{align}
Moreover, $f$ is monotone and submodular, i.e., for any $S \subseteq T \subset X$ and $x \notin T$,
\begin{align}
\textrm{Monotone: } & \quad f_W(T \cup \{x\})  \geq f_W(T) \\
\textrm{Submodular: } & \quad f_W(S \cup \{x\}) - f_W(S) \geq f_W(T \cup \{x\}) - f_W(T) \ .
\end{align}
\end{proposition}
\begin{proof}
The monotonicity of $f_W$ is clear. For the submodularity, let $S \subseteq T \subset X$. For any $u \notin T$,
\begin{align}
& (f_W(S \cup \{u\}) - f_W(S)) - (f_W(T \cup \{u\}) - f_W(T)) \\
&= \sum_{y} \left( \max_{x \in S \cup \{u\}} W(y|x)  - \max_{x \in S} W(y|x)
- (\max_{x \in T \cup \{u\}} W(y|x) - \max_{x \in T} W(y|x) ) \right) \ .
\end{align}
For each $y \in Y$, we distinguish two cases. If $\max_{x \in T \cup \{u\}} W(y|x) = W(y|u)$, in which case the expression reduces to $\max_{x \in T} W(y|x) - \max_{x \in S} W(y|x) \geq 0$. The second case is that the maximum is achieved in $T$, i.e., $\max_{x \in T \cup \{u\}} W(y|x) = W(y|x)$ for some $x \in T$. In this case $\max_{x \in T \cup \{u\}} W(y|x) - \max_{x \in T} W(y|x) = 0$ and the above expression is clearly non-negative by monotonicity of $f$.

We now show that $\channel(W, k) \geq \frac{1}{k}\max_{S \subseteq X, |S| \leq k} f_W(S)$. For this, choose $S \subseteq X$ of size $\ell \leq k$ (as $f$ is monotone, we can assume that maximum is attained for some $S$ of size exactly $\min(k, |X|)$). Then arbitrarily order the elements in $S = \{x_1, \dots, x_{\ell}\}$. Define $e(x_{i}|i) = 1$ for all $i \in [\ell]$ and set all the other variables $e(.|i)$ for $i \in [\ell]$ to zero. For $i \in \{\ell+1, \dots, k\}$, set $e(x|i)$ in an arbitrary way that satisfies the normalization constraint. Then, for any $y$ define $m(y)$ to be $i \in [\ell]$ such that $W(y|x_i)$ is the maximum (in case of multiple $i$'s with the same maximum value, choose the smallest $i$). Then set $d(m(y)|y) = 1$ for all $y \in Y$ and zero for all other entries in $d$. Clearly $e$ and $d$ satisfy the constraints. Moreover,
\begin{align}
\frac{1}{k} \sum_{i\in [k],x,y} W(y|x) e(x|i) d(i|y) &\geq \frac{1}{k} \sum_{i \in [\ell],y} W(y|x_i) d(i|y) \\
&= \frac{1}{k} \sum_{y} W(y|x_{m(y)}) \\
&= \frac{1}{k} \sum_{y} \max_{x \in S} W(y|x)\ ,
\end{align}
which leads to the desired result.

For the other direction, if we define $x_i$ to be the symbol that maximizes $\sum_{y} W(y|x) d(i|y)$, we have
\begin{align}
\sum_{i,x} e(x|i) \sum_{y} W(y|x) d(i|y) &\leq \sum_{i} \max_{x} \sum_{y} W(y|x) d(i|y) \\
&= \sum_{i} \sum_{y} W(y|x_i) d(i|y) \\
&\leq \sum_{y} \max_{i \in [k]} W(y|x_i) \\
&\leq \max_{S \subseteq X, |S| \leq k} \sum_{y} \max_{x \in S} W(y|x) \ .
\end{align}
\end{proof}

It follows from the proof above that any algorithm computing an optimal $S$ in~\eqref{eq:s_submodular} can be easily transformed into an algorithm computing optimal encoding and decoding functions $e, d$ in~\eqref{eq:channel_classical}, and vice versa. For this reason, we interchangeably talk about the code $S$ and the encoding-decoding pair $(e,d)$.

We note that the expression~\eqref{eq:s_submodular} is well-known in information theory and comes from taking a maximum likelihood decoder, see e.g.,~\cite[Section III]{VCFM14}. It is also worth observing that $\log f_W(S) = I_{\infty}(U_{S} : Y)$ where the joint distribution of $(U_S, Y)$ is defined by $P_{U_S Y}(x,y) = \frac{1}{|S|} W(y|x)$ if $x \in S$ and zero otherwise. $I_{\infty}$ is the mutual information of order $\infty$ (see~\cite{Ver15} for a discussion of $\alpha$-mutual information).

Using the notable result of~\cite{NWF78}, the formulation in \eqref{eq:s_submodular} immediately shows that the quantity $\channel(W,k)$ can be approximated efficiently within a factor of $(1-e^{-1})$. In fact, this can be achieved with a very simple greedy algorithm. Starting with set $S_{0} = \emptyset$, subset $S_{\ell+1} \subseteq X$ is constructed from subset $S_{\ell} \subseteq X$ by adding an element $x_{\ell+1}$ that maximizes $f_W(S_{\ell} \cup \{ x_{\ell+1} \})$, so that $S_{\ell+1} = S_{\ell} \cup \{x_{\ell+1}\}$. Let $\channel^{\greedy}(W,k)$ be the success probability obtained by the greedy algorithm on input $W$ and $k$. With the notation used above and the function $f_W$ defined in \eqref{eq:s_submodular}, we have $\channel^{\greedy}(W,k) = \frac{1}{k} f_{W}(S_{k})$; here $S_k$ is the size-$k$ subset computed by the greedy algorithm. 

\begin{corollary}
\label{cor:greedy-opt}
For any channel $W$ and any $k$,
\begin{align}
(1-e^{-1})\channel(W,k) \leq \channel^{\greedy}(W, k) \leq \channel(W,k) \ .
\end{align}
Moreover, if there is a polynomial-time algorithm that takes as input $W$ and $k$ and returns a number $\mathrm{Alg}(W, k)$ satisfying for some $\e > 0$ and all inputs $W, k$ the inequality 
$(1-e^{-1} + \eps)\channel(W,k) \leq \mathrm{Alg}(W, k) \leq \channel(W,k)$, then {\rm P} = {\rm NP}.
\end{corollary}
\begin{proof}
The fact that the greedy algorithm provides a $(1-e^{-1})$-approximation algorithm for $\channel$ follows directly from~\cite{NWF78}. In Section~\ref{sect:ns}, we provide a proof of a strengthening of this result.

For the hardness of the problem, we use the hardness of the maximum-$k$-coverage problem~\cite{Fei98}. In the maximum-$k$-coverage problem, we are given a collection of sets $\{T_x\}_{x \in X}$ of elements in $Y$ (i.e., $T_x \subseteq Y$ for each $x \in X$) and the objective is to find a subset $S \subseteq X$ of size $k$ such that $|\cup_{x \in S} T_x|$ is as large as possible. Feige~\cite{Fei98} showed that this problem is hard to approximate with a factor better than $(1-e^{-1})$. In fact, as highlighted in~\cite{Fei02}, the problem is still hard if the sets $T_x$ all have the same size, call it $d$. Given such an instance we define the following channel $W(y|x) = \frac{1}{d}$ if $y \in T_x$ and $0$ otherwise. Then for any choice $S \subseteq X$, we have $| \cup_{x \in S} T_x | = \sum_{y} \max_{x \in S} W(y|x) \cdot d = d \cdot f_W(S)$. This shows the desired result.
\end{proof}

\section{Channel coding with non-signaling boxes}
\label{sect:ns}
A key motivation for this work is to understand the advantage that can be achieved when the sender and the receiver share additional resources that are by themselves useless for communication. Such resources are commonly called non-signaling boxes~\cite{Pop14}. The simplest example of a non-signaling box is a device providing shared randomness between the sender and the receiver. It is quite simple to see that allowing the encoder $e$ and decoder $d$ to depend on some shared randomness will not affect the value of~\eqref{eq:channel_classical}. However, if the sender and the receiver share entanglement, we know that for some channels, a success probability $\channel^{\quantum}(W,k)$ that exceeds $\channel(W,k)$ can be achieved~\cite{PLMKR11,CLMW10}. This is analogous to a Bell inequality violation~\cite{BCPSW14}, or in other words the fact that the entangled value of a 2-prover 1-round game can be larger than the classical value.\footnote{In fact, the problem of optimal channel coding that we are studying can be seen as a kind of game. The input of the sender is the label $i$ of the message and his output is an element $x \in X$. The input of the receiver is some $y \in Y$ and his input is a label $j$ of some message. The difference is that the way one would normalize for a game is different than the way we do it here. Also, in our setting, the referee's output is not necessarily $0$ or $1$ but rather a utility for each input-output pair that depends on the channel probabilities $W(y|x)$.}

A natural question to ask here is how much can entanglement (or, in general, a non-signaling resource) between the sender and receiver help for reliable transmission. For example, is there a choice of channel for which the success probability with entanglement is close to $1$ and without entanglement is very small? Our main result shows that this cannot be the case and that the ratio $\frac{\channel(W,k)}{\channel^{\quantum}(W,k)} \geq 1 - e^{-1}$. As $\channel^{\quantum}(W,k)$ does not seem to be easy to analyze,\footnote{We only know of a hierarchy of semidefinite programs that converges to this value~\cite{BFS15}.} it is helpful to consider even more general resources between the sender and the receiver. Allowing for any non-signaling box between the sender and the receiver leads to a very natural linear programming (LP) relaxation of~\eqref{eq:channel_classical}. 
\begin{align}\label{eq:channel_lp}
\begin{aligned}
\channel^{\NS}(W, k)\eqdef\;& \underset{r_{x,y}, p_x}{\text{maximize}}
& & \frac{1}{k} \sum_{x,y} W(y|x) r_{x,y} \\
& \text{subject to}
& &\sum_{x} r_{x,y} \leq 1 \quad \forall y \in Y \\
& 
& &\sum_{x} p_{x} = k \\
& 
& &  r_{x,y} \leq p_x \quad \forall (x,y) \in X \times Y\\ 
&
& & 0 \leq r_{x,y}, p_{x} \leq 1  \quad \forall (x,y) \in X \times Y\,.
\end{aligned}
\end{align}
Matthews~\cite{Mat12} showed that $\channel^{\NS}(W,k)$ corresponds to the maximum success probability that can be achieved if the sender and the receiver are allowed to use any additional non-signaling resource (see also~\cite[Section II.B]{Pol13} for another explanation of this fact.) This means that they can use any additional box as long as it does not allow the sender and the receiver to send information to each other; see~\cite{Pop14} and references therein for a review on the study of non-signaling boxes. For the convenience of the reader, we repeat the proof of~\cite{Mat12} in Appendix~\ref{sec:lp_ns}.\footnote{Note that we also have a small modification compared to the LP of~\cite{Mat12} in that we require $p_x \leq 1$. We can safely add this constraint provided $k \leq |X|$, which we assume throughout this section.}

As pointed out in~\cite{Mat12}, a form of the LP~\eqref{eq:channel_lp} is widely known in the information theory literature as the Polyanskiy-Poor-Verd\'u meta-converse~\cite{PPV10}. The PPV meta-converse gives an upper bound for the number of messages that can be sent through a channel in terms of some hypothesis test; a connection which also appeared in~\cite{HN03}. In Appendix~\ref{sec:lp_ht}, we basically repeat the argument of~\cite{Mat12} to show how to interpret the LP~\eqref{eq:channel_lp} in terms of a hypothesis test.

Our main result shows that the LP relaxation~\eqref{eq:channel_lp} cannot be too far from the maximum success probability in~\eqref{eq:channel_classical}.
\begin{theorem}
\label{thm:ns-advantage}
Let $W$ be a channel. Then, for any integers $\ell, k \in \{1, \dots, |X|\}$ we have:
\begin{align}
\label{eq:ns-advantage}
\channel(W,\ell) &\geq \frac{k}{\ell} \left( 1 - \left(1 - \frac{1}{k}\right)^{\ell} \right) \channel^{\NS}(W, k) \ .
\end{align}
More precisely, this can be achieved via the greedy algorithm
\begin{align}
\label{eq:ineq-greedy-lp}
\channel^{\greedy}(W,\ell) &\geq \frac{k}{\ell} \left( 1 - \left(1 - \frac{1}{k}\right)^{\ell} \right) \channel^{\NS}(W, k) \ ,
\end{align}
or via random coding
\begin{align}
\label{eq:ineq-random-lp}
\exc{S}{\frac{1}{\ell} f_{W}(S)} &\geq \frac{k}{\ell} \left( 1 - \left(1 - \frac{1}{k}\right)^{\ell} \right) \channel^{\NS}(W, k) \ ,
\end{align}
where $S \subseteq X$ is obtained by choosing $\ell$ elements of $X$ independently according to the distribution $\{\frac{p_{x}}{k}\}_{x \in X}$, where $p_x$ is an optimal solution in \eqref{eq:channel_lp} and $f_{W}$ is as defined in Proposition~\ref{prop:s_submodular}.
\end{theorem}
Figure~\ref{fig:s_of_l} gives an illustration for the statement of the theorem.
	\begin{figure}
	\centering
		\includegraphics[width=10cm]{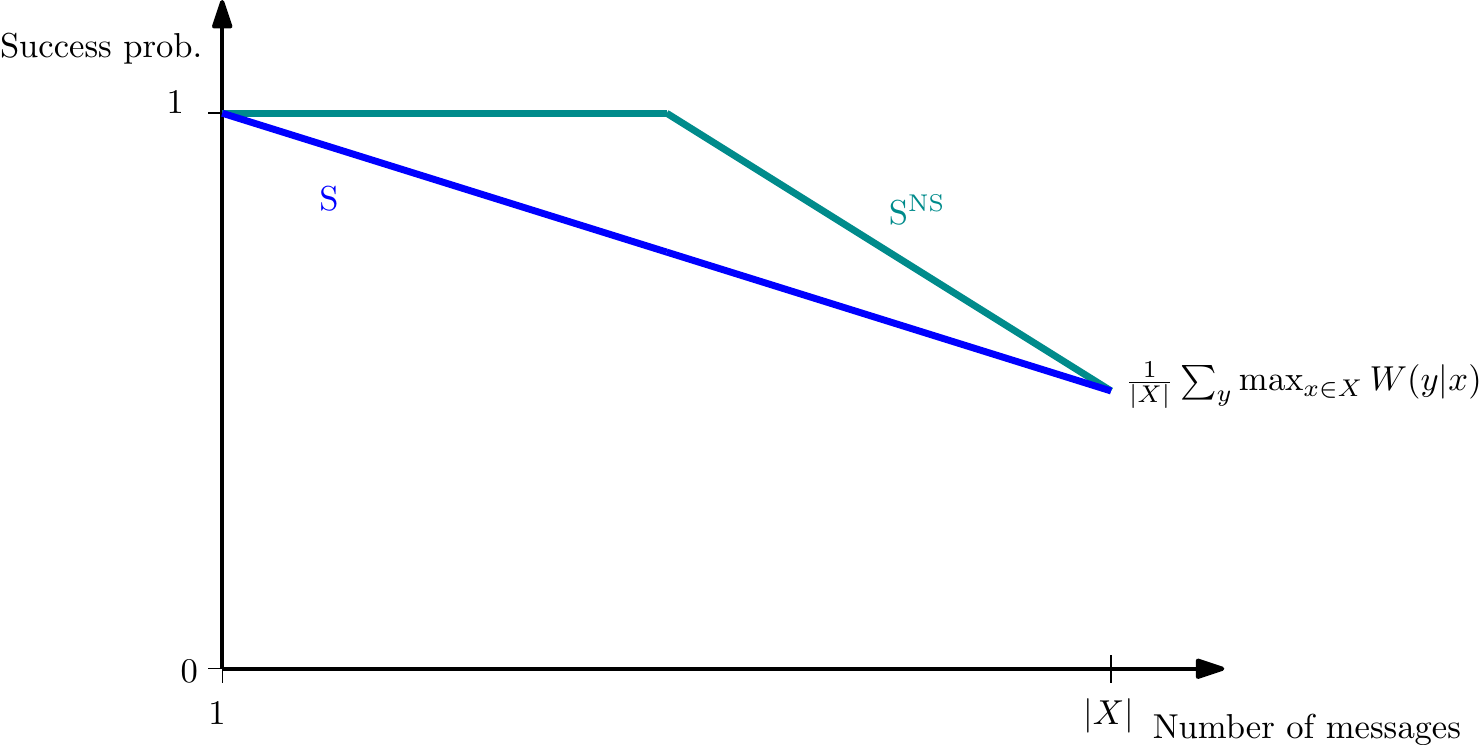}
		\caption{This plot illustrates the maximum success probability as a function of the number of messages to be sent. The top curve corresponds to the maximum success probability if any non-signaling boxes between the sender and the receiver is allowed. The bottom one corresponds to the setting where the sender and the receiver do not share any additional resources. Theorem~\ref{thm:ns-advantage} states that for any value $\ell$ of the number of messages, the ratio between the non-signaling success probability and the usual success probability is at most $1-e^{-1}$. It also gives a way of comparing the two values for different number of messages to be sent.}
		\label{fig:s_of_l}
	\end{figure}
%
%
%
%
%
%
%

\paragraph{Comments on the proof}
As mentioned in the theorem, we give two proofs of this result. The first one by relating the performance of the greedy algorithm to the linear program~\eqref{eq:channel_lp}. The $\ell = k$ case of this theorem can be proved by using a result of~\cite{CFN77} relating the performance of the greedy algorithm and the linear programming relaxation for the location problem. In fact, the expression in~\eqref{eq:s_submodular} shows that the optimal channel coding problem can be written as a location problem. To obtain the tradeoff between success probability and number of messages, we use the observation about the greedy algorithm that can be found in~\cite[Theorem 1.5]{KG12}. A complete proof of this theorem appears in Section~\ref{sec:proof-greedy}. The second proof proceeds by standard randomized rounding techniques and can be found in Section~\ref{sec:proof-random}

Inequality \eqref{eq:ineq-greedy-lp} (for $\ell = k$) together with the fact that both $\channel^{\greedy}(W, k)$ and $\channel^{\NS}(W,k)$ are computable in polynomial time might seem at first in contradiction with the NP-hardness in Corollary~\ref{cor:greedy-opt}. But this is not the case because it is unclear how to use the linear program $\channel^{\NS}(W,k)$ to obtain a lower bound on $\channel(W, k)$ that is better than the greedy algorithm.
Another consequence of~\eqref{eq:ineq-greedy-lp} is that it proves that the greedy algorithm gives a $(1-e^{-1})$ approximation even for the maximum success probability $\channel^{\quantum}(W,k)$ using entanglement.

Another relevant observation concerning the proof is that we can in fact obtain a multiplicative bound between centered quantities. As $\channel(W,k) \geq \frac{1}{k}$ for any channel $W$ (the decoder can just randomly guess a message), it might be useful to consider the ratio $\channel(W,k) - \frac{1}{k}$ to $\channel^{\NS}(W,k) - \frac{1}{k}$, in particular for small values of $k$. A simple modification in the proof leads to the bound:
\begin{align}
\label{eq:centered_ineq}
\channel(W,k) - \frac{1}{k} \geq \left(1 - \left(1 - \frac{1}{k}\right)^{k-1}\right) \left(\channel^{\NS}(W,k) - \frac{1}{k} \right) \ .
\end{align} 
This inequality generalizes the bounds obtained by Hemenway et al.~\cite{HMSW13}, who considered the case $k=2$.

\paragraph{Comments on the ratio}
The pre-factor in the right hand side of~\eqref{eq:ns-advantage} can be simplified using the following inequalities
\begin{align}
\frac{k}{\ell} \left( 1 - \left(1 - \frac{1}{k}\right)^{\ell} \right)
\geq \frac{k}{\ell} \left( 1 - e^{-\ell/k} \right)
\end{align}
and when $\ell/k \ll 1$, a good approximation is given by the following inequality
\begin{align}
\label{eq:expansion}
\frac{k}{\ell} \left( 1 - e^{-\ell/k} \right) \geq 1 - \frac{\ell}{2k} \ .
\end{align}
The bound shows, for example, that if it is possible to send $\log k$ bits of information with a success probability $\channel^{\NS}(W, k) = 1-\e$ using non-signaling boxes, then it is also possible to send $\log k - 10$ bits of information with a success probability of at least $0.998 \cdot (1-\e)$ without using any additional resources.  
We show in Section~\ref{sec:tightness} that the bound~\eqref{eq:ns-advantage} is tight.


\paragraph{Multiple independent uses of a channel}
It has been known for a long time that asymptotically for $n$ channel uses with $n \to \infty$, entanglement cannot increase the capacity $\capacity(W)$ of a noisy classical channel $W$~\cite{BSST99,BBCJPW93}. In fact, this is easily recovered from Theorem~\ref{thm:ns-advantage}. Let $R$ be a rate achievable using a non-signaling box, i.e., $\channel^{\NS}(W^{\otimes n}, R^n) \to 1$ as $n \to \infty$. We will show that $R \leq \capacity(W)$. In fact, let $\delta > 0$. Then using Theorem~\ref{thm:ns-advantage} with the 
\begin{align}
\channel(W^{\otimes n}, (R(1-\delta))^n) &\geq \left(1 - \frac{1}{2} \left(\frac{R(1-\delta)}{R}\right)^n\right) \channel^{\NS}(W^{\otimes n}, R^n) \ ,
\end{align}
and thus $\channel(W^{\otimes n}, R^n(1-\delta)^n) \to 1$, which shows that $R(1-\delta)$ is an achievable rate for the channel $W$ (without using any non-signaling boxes). By definition of the capacity $\capacity(W)$, this means that $R(1 - \delta) \leq \capacity(W)$. As this holds for any $\delta$, this shows that $R \leq \capacity(W)$.


\subsection{Proof of Theorem~\ref{thm:ns-advantage} via the greedy algorithm}
\label{sec:proof-greedy}
In comparison with Corollary~\ref{cor:greedy-opt}, we prove a stronger result relating the performance of the greedy algorithm to the value of the linear programming relaxation. For that, it is useful to define the following extension of the function $f_W$ defined in~\eqref{eq:s_submodular} to fractional vectors $p \in [0,1]^{|X|}$,
\begin{align}\label{eq:f_fractional}
\begin{aligned}
f_{W}(p) \eqdef\;& \underset{r_{x,y}}{\text{maximize}}
& & \sum_{x,y} W(y|x) r_{x,y} \\
& \text{subject to}
& &\sum_{x} r_{x,y} \leq 1 \quad \forall y \in Y \\
& 
& & 0 \leq r_{x,y} \leq p_x \quad \forall (x,y) \in X \times Y\,.
\end{aligned}
\end{align}
With this notation, we can write 
\begin{align}
\label{eq:ns_f}
\channel^{\NS}(W, k) = \frac{1}{k} \max\limits_{\substack{p_x \geq 0 \\ \sum_{x} p_x = k}} f_W(p) \ .
\end{align}
The following lemma is crucial in proving Theorem~\ref{thm:ns-advantage}.
\begin{lemma}
\label{lem:fractional-greedy}
Let $S \subseteq X$ and vector $p \in [0,1]^{|X|}$ such that $\sum_{x} p_x = k$. Then 
\begin{align*}
f_W(p) \leq f_W(S) + k(\max_{x \in X} f_W(S \cup \{x\}) - f_W(S) )
\end{align*}
\end{lemma}

\begin{proof}
Define for any $x \in X$, $q_x = \max\{p_x, 1_{x \in S},0\}$. Note that $q_x \geq p_x$ for all $x$, and so $f_W(p) \leq f_W(q)$. We aim to show the stronger statement $f_W(q) \leq f_W(S) + k  (\max_{x \in X} f_W(S \cup \{x\}) - f_W(S) )$. Let $r_{x,y}$ correspond to an optimal solution for the program of $f_{W}(q)$, in particular $f_{W}(q) = \sum_{x,y} W(y|x) r_{x,y}$. We have
\begin{align}
&f_W(q) - f_W(S) = \sum_{x,y} W(y|x) r_{x,y} - \sum_{y} \max_{x' \in S} W(y|x') \\
		&= \sum_{y} \left(\sum_{x} W(y|x) r_{x,y} - \max_{x' \in S} W(y|x') \right) \\
		&\leq \sum_{y} \left(\sum_{x \in S} W(y|x) r_{x,y} - (\sum_{x \in S} r_{x,y}) \max_{x' \in S} W(y|x')  + \sum_{x \notin S} W(y|x) r_{x,y} - (\sum_{x \notin S} r_{x,y}) \max_{x' \in S} W(y|x') \right) \ ,
		\end{align}
		using the fact that $\sum_{x} r_{x,y} \leq 1$. Now observe that $\sum_{x \in S} W(y|x) r_{x,y} - (\sum_{x \in S} r_{x,y}) \max_{x' \in S} W(y|x') \leq 0$ and thus
		\begin{align}
		f_W(q) - f_W(S) &\leq \sum_{y} \left(\sum_{x \notin S} W(y|x) r_{x,y} - (\sum_{x \notin S} r_{x,y}) \max_{x' \in S} W(y|x') \right) \\
		&= \sum_{x \notin S} \sum_{y} r_{x,y} (W(y|x) - \max_{x' \in S} W(y|x')) \\
		&\leq \sum_{x \notin S} \sum_{y \in \Gamma(x)} r_{x,y} (W(y|x) - \max_{x' \in S} W(y|x'))
		\label{eq:diff-xstar} 
		\ ,
\end{align}
where $\Gamma(x) = \{ y : W(y|x) - \max_{x' \in S} W(y|x') > 0 \}$.
For $x \notin S$, we have $q_{x} = p_{x}$ and thus $r_{x,y} \leq p_{x}$. As a result,
\begin{align}
f_W(q) - f_W(S) &\leq \sum_{x \notin S} p_{x} \sum_{y \in \Gamma(x)} (W(y|x) - \max_{x' \in S} W(y|x')) \\
&\leq k \cdot \sum_{y \in \Gamma(x^*)} ( W(y|x^*) - \max_{x' \in S} W(y|x') ) \ \label{eq:final-diff},
\end{align}
where $x^* \notin S$ maximizes the quantity $\sum_{y \in \Gamma(x)} (W(y|x) - \max_{x' \in S} W(y|x'))$. Now observe that by definition of $\Gamma(x^*)$, we have
\begin{align}
f_W(S \cup \{x^*\}) - f_W(S) &= \sum_{y \in \Gamma(x^*)} ( W(y|x^*) - \max_{x' \in S} W(y|x') ) .
\end{align}
Combining this with \eqref{eq:final-diff}, we get the desired result.
\end{proof}


\begin{proof}[of Theorem~\ref{thm:ns-advantage}, Eq. \eqref{eq:ineq-greedy-lp}]
Using Lemma~\ref{lem:fractional-greedy}, we can apply the framework of~\cite{NWF78} for analyzing the performance of the greedy algorithm. Recall the notation introduced for the greedy algorithm:  starting from $S_{0} = \emptyset$, $S_{\ell+1}$ is constructed from $S_{\ell}$ by adding an element $x_{\ell+1}$ that maximizes $f_W(S_{\ell} \cup \{ x_{\ell+1} \})$, so that $S_{\ell+1} = S_{\ell} \cup \{x_{\ell+1}\}$.
Fix an integer $i_0 \in \{0, \dots, |X|\}$ (think of $i_0 = 0$, but it will also be useful to choose $i_0 = 1$). We prove by induction on $\ell$ that for any $\ell \geq i_0$
\begin{align}
\label{eq:induction}
\max_{p} f_W(p) - f_W(S_{\ell}) \leq \left(1 - \frac{1}{k} \right)^{\ell-i_0} \left(\max_{p} f_W(p) - f_W(S_{i_0}) \right) \ ,
\end{align}
where the maximizations are taken over all $p$ such that $p_x \geq 0$ and $\sum_{x} p_x = k$.
The base case $\ell = i_0$ is clear. Using Lemma~\ref{lem:fractional-greedy} together with the fact that $f_W(S_{\ell+1}) = \max_{x^* \in X} f_W(S_{\ell} \cup \{x^*\})$ gives
\begin{align}
\max_p f_W(p) - f_W(S_{\ell}) \leq k \cdot (f_W(S_{\ell+1}) - f_W(S_{\ell}) ) \ .
\end{align}
Rearranging the terms we see that
\begin{align}
\max_p f_W(p) - f_W(S_{\ell+1}) &\leq \left(1 - \frac{1}{k}\right) \left( \max_p f_W(p) - f_W(S_{\ell}) \right) \\
&\leq \left(1-\frac{1}{k}\right)^{\ell + 1 - i_0} \left( \max_p f_W(p) - f_W(S_{i_0}) \right) \ ,
\end{align}
using the induction hypothesis. For $i_0=0$, this gives can be written as
\begin{align}
f_W(S_\ell) \geq \left( 1 - \left(1-\frac{1}{k}\right)^{\ell} \right) \max_p f_W(p) \ .
\end{align}

Dividing by $\ell$ and using~\eqref{eq:ns_f}, we have
\begin{align}
\channel(W,\ell) \geq \channel^{\greedy}(W, \ell) = \frac{1}{\ell} f_W(S_{\ell}) &\geq \frac{1}{\ell} \left( 1-\left(1- \frac{1}{k}\right)^{\ell} \right) \max_{p} f_W(p) \\ 
& = \frac{k}{\ell} \left( 1-\left(1- \frac{1}{k}\right)^{\ell} \right) \channel^{\NS}(W, k) \ .
\end{align}
This concludes the proof of the theorem. To prove~\eqref{eq:centered_ineq}, take $i_0 = 1$, observe that $f_W(S_1) = 1$ and set $\ell = k$, \eqref{eq:induction} becomes
\begin{align}
\frac{\max_p f_W(p) - f_W(S_{k})}{\max_p f_W(p) - 1} \leq \left( 1 - \frac{1}{k} \right)^{k-1} \ .
\end{align}
Rearranging the terms, we get the inequality~\eqref{eq:centered_ineq}.
\end{proof}

\subsection{Proof of Theorem~\ref{thm:ns-advantage} via random coding}
\label{sec:proof-random}

\begin{proof}[of Theorem~\ref{thm:ns-advantage}, Eq. \eqref{eq:ineq-random-lp}]
Recall that $S$ is obtained by choosing $\ell$ independent samples from the distribution $\{\frac{p_{x}}{k}\}_{x \in X}$. We have
\begin{align*}
\exc{S}{\frac{1}{\ell} f_{W}(S)} &= \frac{1}{\ell} \sum_{y} \exc{S}{\max_{x \in S} W(y|x)} \ .
\end{align*}
We will show that for any $y \in Y$, $\exc{S}{\max_{x \in S} W(y|x)} \geq \sum_{x \in X} \left( 1 - \left(1 - \frac{1}{k}\right)^{\ell} \right) W(y|x) r_{x,y}$.

In order to show this, we order the inputs $\{x_1, \dots, x_{|X|}\}$ in $X$ by decreasing order of $W(y|x)$, so that $W(y|x_1) \geq W(y|x_2) \geq \dots \geq W(y|x_{|X|})$. We may then write
\begin{align*}
\exc{S}{\max_{x \in S} W(y|x)} &= \sum_{i=1}^{|X|} \prc{S}{x_{1} \notin S \cap \dots \cap x_{i-1} \notin S \cap x_{i} \in S} W(y|x_i) \ .
\end{align*}
Note that we used the convention that $x_0 \in S$ always holds. Observe that 
\begin{align*}
\prc{S}{x_{1} \notin S \cap \dots \cap x_{i-1} \notin S \cap x_{i} \in S} = \prc{S}{x_1 \in S \cup \dots \cup x_{i} \in S} - \prc{S}{x_1 \in S \cup \dots \cup x_{i-1} \in S} \ .
\end{align*}
Thus, we obtain
\begin{align}
\label{eq:random_code_term_y}
\exc{S}{\max_{x \in S} W(y|x)} &= \sum_{i=1}^{|X|-1} \prc{S}{x_{1} \in S \cup \dots \cup x_{i} \in S} (W(y|x_i) - W(y|x_{i+1})) + W(y|x_{|X|}) \ .
\end{align}
Now we can find a lower bound on $\prc{S}{x_{1} \in S \cup \dots \cup x_{i} \in S}$ using the variables $p_{x}$ and $r_{x,y}$ of the linear program. In fact, 
\begin{align*}
\prc{S}{x_{1} \in S \cup \dots \cup x_{i} \in S} &= 1 - \left(1 - \frac{p_{x_1} + \dots + p_{x_i}}{k}\right)^{\ell} \\
&\geq 1 - \left(1 - \frac{r_{x_1,y} + \dots + r_{x_i,y}}{k}\right)^{\ell} \\
&\geq \left( 1 - \left(1 - \frac{1}{k}\right)^{\ell} \right) (r_{x_1,y} + \dots + r_{x_i,y}) \ ,
\end{align*}
where the first inequality comes from the constraint $r_{x,y} \leq p_{x}$ and the last inequality from the constraint $\sum_{x} r_{x,y} \leq 1$ and the concavity of the function $z \mapsto 1 - (1 - \frac{z}{k})^{\ell}$. Going back to \eqref{eq:random_code_term_y}, we obtain
\begin{align}
\exc{S}{\max_{x \in S} W(y|x)} &\geq \left( 1 - \left(1 - \frac{1}{k}\right)^{\ell} \right)  \sum_{i=1}^{|X|-1} (r_{x_1,y} + \dots + r_{x_i,y}) (W(y|x_i) - W(y|x_{i+1})) + W(y|x_{|X|}) \\
&\geq \left( 1 - \left(1 - \frac{1}{k}\right)^{\ell} \right)  \sum_{i=1}^{|X|} r_{x_i,y} W(y|x_i) \ ,
\end{align}
where we used again in the second inequality the constraint that $\sum_{x} r_{x,y} \leq 1$. This proves the claimed inequality and thus by summing over $y \in Y$, we obtain
\begin{align*}
\exc{S}{\frac{1}{\ell} \sum_{y} \max_{x \in S} W(y|x)} \geq \frac{k}{\ell} \left( 1 - \left(1 - \frac{1}{k}\right)^{\ell} \right) \channel^{\NS}(W, k) \ .
\end{align*}
\end{proof}

\subsection{Tightness of Theorem~\ref{thm:ns-advantage}}
\label{sec:tightness}

We now prove that the result shown in Theorem~\ref{thm:ns-advantage} is tight using a simple family of graphs proposed in~\cite{AS04_pipage}.
Consider the following channel for $k,t \geq 1$ integers. The input alphabet $X$ is composed of $n \eqdef kt$ symbols and the output alphabet $Y$ is composed of $\binom{n}{t}$ symbols that we interpret as subsets of $X$ of size $t$. On input $x$, the output of the channel is a randomly chosen $y$ such that $x \in y$. 
\begin{align}
W(y|x) = \left\{ \begin{array}{cl}
\frac{1}{\binom{n-1}{t-1}} & \text{ if $x \in y$} \\
0 & \text{ otherwise.}
\end{array}
\right.
\end{align}
Note that, interestingly, the case $k = t = 2$ is exactly the channel that is studied in~\cite{PLMKR11}, in which it is experimentally demonstrated that entanglement assistance can help in improving the success probability for sending a bit over this channel.

We first show that $\channel^{\NS}(W,k) = 1$. For this let $p_x = \frac{k}{n}$ for all $x \in X$ and $r_{x,y} = \frac{k}{n}$ if $x \in y$ and $r_{x,y} = 0$ otherwise. We have $\sum_{x} r_{x,y} = t \frac{k}{n} = 1$. Moreover,
\begin{align}
\frac{1}{k} \sum_{x,y} W(y|x) r_{x,y} = \frac{1}{k} \sum_{y} \sum_{x \in y} \frac{1}{\binom{n-1}{t-1}} \frac{k}{n} = \frac{1}{k} \binom{n}{t} \cdot t \frac{1}{\binom{n-1}{t-1}} \frac{k}{n} = 1 \ .
\end{align}
Using the symmetry of the channel it is simple to determine $\channel(W, \ell)$ for any $1 \leq \ell \leq n$. In fact, one can see that $f_{W}(S)$ only depends on the size of $S$, let $S_{\ell}$ be any fixed set of size $\ell$:
\begin{align}
\channel(W, \ell) &= \frac{1}{\ell} \max_{|S| = \ell} f_{W}(S) = \frac{1}{\ell} \sum_{y \in Y : y \cap S_{\ell} \neq \emptyset} \frac{1}{\binom{n-1}{t-1}} \ .
\end{align}
So we only need to count the number of subsets $y$ that intersect with the set $S_{\ell}$ of size $\ell$. This number is given by $\binom{n}{t} - \binom{n-\ell}{t}$. Observing that $\binom{n-1}{t-1} = \frac{t}{n} \binom{n}{t}$ and $t/n = 1/k$ we have
\begin{align}
\channel(W, \ell) = \frac{k}{\ell}\left(1 - \frac{\binom{n-\ell}{t}}{\binom{n}{t}} \right) &= \frac{k}{\ell}\left(1 - \frac{\binom{n-t}{\ell}}{\binom{n}{\ell}} \right) \\
&= \frac{k}{\ell}\left(1 - \frac{(n-t) \cdots (n-\ell - t+1)}{n (n-1) \cdots (n-\ell+1)} \right) \\
&= \frac{k}{\ell} \left( 1 - \left(1-\frac{t}{n}\right) \cdots \left( 1 - \frac{t}{n-\ell+1}\right) \right) \ .
\end{align}
From this expression, we see that for example by fixing $\ell$ and $k$ to be constants and letting $t \to \infty$. This expression approaches
\begin{align}
\frac{k}{\ell} \left( 1 - \left(1-\frac{1}{k} \right)^{\ell} \right) \ ,
\end{align}
which exactly matches the bound in Theorem~\ref{thm:ns-advantage}.

\section{Discussion}

The main message of this work is to draw a connection between the study of optimal coding for noisy channels in information theory and algorithmic aspects of submodular maximization. We believe this connection could be fruitful in both directions. As we showed in this paper, techniques developed in the context of submodular maximization---and, in general, approximation algorithms---can have interesting implications when analyzing the problem of optimal channel coding. We believe there are many more relevant applications to be explored. A particular question is whether algorithmic techniques can be helpful in obtaining better finite-blocklength bounds for well-studied channels such as the ones given in~\cite{Pol10}.

\section*{Acknowledgements}

We would like to thank Mario Berta, Shaddin Dughmi, Yury Polyanskiy, Volkher Scholz, Piyush Srivastava, Vincent Tan, Marco Tomamichel as well as the ISIT referees for discussions and comments on this paper.

\appendix

\section{Non-signaling assisted channel coding}
\label{sec:lp_ns}
Suppose that the sender and the receiver have a box shared between them with the following properties. Alice inputs $\alpha$ and receives output $a$ and Bob inputs $\beta$ and receives output $b$. We say that such a box is non-signaling if by itself it is useless for communication. More formally, such a box is described by a conditional probability distribution $P(a,b|\alpha,\beta)$ representing the probability that the outputs are $a$ and $b$ given that the inputs are $\alpha$ and $\beta$. The non-signaling property is then easily formulated as a linear constraint on these numbers: the marginal distribution on $a$ is independent of the input $\beta$ of Bob and the marginal distribution on $b$ is independent of the input $\alpha$ of Alice
\begin{align}
P(a|\alpha, \beta) &\eqdef \sum_{b} P(a,b|\alpha,\beta) = P_A(a | \alpha) \\
P(b|\alpha, \beta) &\eqdef \sum_{a} P(a,b|\alpha,\beta) = P_B(b | \beta) \ ,
\end{align}
for some conditional distributions $P_A$ and $P_B$. In other words, $P(a| \alpha, \beta)$ is the same for all values of $\beta$.
Perhaps the simplest example of a non-signaling box is one that provides shared randomness. In this case $P(a,a'|\alpha, \beta) = \frac{\delta_{a=a'}}{A}$. Also, if Alice and Bob share a quantum state (which could be entangled) and perform measurements that depend on their inputs, the outputs being the measurement outcomes, this also defines a non-signaling box. There are also some distributions that are non-signaling but do not seem to be physically realizable without communication, the most well-known being the Popescu-Rohrlich box~\cite{PR94}.

Now assume that the sender and the receiver share such a box. The sender may give an input depending on the message he wishes to send and use the output he receives from the box to choose the symbol $x \in X$. Similarly, the receiver can give an input $\beta$ that depends on the symbol $y$ he receives and might use the output to decode the message. By encompassing all pre- and post-processing of the sender and the receiver into the box itself, we can assume that Alice inputs the message $i$ into the box and the output of the box is exactly the input to the channel. Similarly, Bob inputs $y$ into the box and receives $j$, a candidate for the sent message.

Given this definition, we can naturally define the non-signaling success probability as
\begin{align}\label{eq:channel_lp_def}
\begin{aligned}
\channel^{\NS}(W, k)\eqdef\;& \underset{P(x,j|i,y), P_A, P_B}{\text{maximize}}
& & \frac{1}{k} \sum_{x,y,i} W(y|x) P(x,i|i,y) \\
& \text{subject to}
& &\sum_{j} P(x,j|i,y) = P_{A}(x|i) \quad \forall i,x,y \in [k] \times X \times Y \\
& 
& &\sum_{x} P(x,j|i,y) = P_{B}(j|y)  \quad \forall j,i,y \in [k] \times [k] \times Y\\
& 
& &  \sum_{x,j} P(x,j|i,y) = 1 \quad \forall i,y \in [k] \times Y\\ 
&
& & 0 \leq P(x,j|i,y)  \quad \forall j,i,x,y \in [k] \times [k] \times X \times Y\ .
\end{aligned}
\end{align}

We here prove that this linear program and the one in~\eqref{eq:channel_lp} have the same value. 

Given a box $P$, let us construct a feasible solution for~\eqref{eq:channel_lp}. Let $r_{x,y} = \sum_{i} P(x,i|i,y)$ and $p_x = \sum_i P_A(x|i)$. Then clearly $\sum_{x} r_{x,y} = \sum_{i} P_{B}(i|y) = 1$ and $\sum_x p_x = \sum_{x,i} P_A(x|i) = k$. We still need to show that we can assume $p_x \leq 1$. For this, define $p'_x = \min\{p_x, 1\}$. As $r_{x,y} \leq 1$, we still have $r_{x,y} \leq p'_x$. But the sum $\sum_{x} p'_x$ might be less than $k$. But as $k \leq |X|$, there exists $\bar{p}_x \geq p'_x$ while keeping $\bar{p}_x \leq 1$ satisfying $\sum_{x} \bar{p}_x = k$. The pair $(r_{x,y}, \bar{p}_x)$ is then a feasible solution for~\eqref{eq:channel_lp} with objective function equal to $\frac{1}{k} \sum_{x,y,i} W(y|x) P(x,i|i,y)$.

For the other direction, define 
\begin{align}
P(x,j|i,y) = \left\{ \begin{array}{cl}
\frac{r_{x,y}}{k} & \text{ if } i = j \\
\frac{p_x - r_{x,y}}{k(k-1)} & \text{ if } i \neq j \ .
\end{array}
\right .
\end{align}
It is simple to see that this distribution defines a non-signaling box.

\section{Interpreting the LP in terms of hypothesis testing}
\label{sec:lp_ht}

Consider two distributions $P$ and $Q$ over the set $Z$. Given a sample from either $P$ or $Q$, we wish to determine which distribution generated the sample. One can define a randomized test $T : Z \to [0,1]$ where $T(z)$ is the probability of declaring the distribution to be $P$. An important quantity that is studied in statistical hypothesis testing is the smallest probability to falsely outputting $P$ among all tests that correctly identify $P$ with probability at least $\alpha$: More precisely,
\begin{align}
\beta_{\alpha}(P, Q) = \min_{T : Z \to [0,1], \sum_{z} P(z) T(z) \geq \alpha} \sum_{z} Q(z) T(z) \ .
\end{align}
We will see that the quantity $\channel^{\NS}(W,k)$ is related to $\beta_{\alpha}(P, Q)$ for distributions $P$ and $Q$ constructed from the channel $W$. More specifically, let $\mu \in \dist(X)$ be a distribution on inputs $X$ and $\nu \in \dist(Y)$ be a distribution on outputs $Y$. Then let the distribution $\mu \cdot W$ on $X \times Y$ be defined by probabilities $\mu(x) W(y|x)$ and $\mu \cdot \nu$ be the distribution defined by probabilities $\mu(x) \nu(y)$. Then the maximum success probability $\channel^{\NS}(W,k)$ can be interpreted as a distance between the product distribution $\mu \cdot \nu$ and the distribution induced by the channel $\mu \cdot W$ in the following sense: 
\begin{proposition}
\begin{align}
1 - \channel^{\NS}(W, k) = \min_{\mu \in \dist(X)} \max_{\nu \in \dist(Y)} \beta_{1-1/k}(\mu \cdot \nu, \mu \cdot W) \ .
\end{align}
\end{proposition}
\begin{proof}
Given a feasible solution $p_x$ and $r_{x,y}$ for \eqref{eq:channel_lp}, we define a distribution $\mu(x) = p_x/k$ and a test $T(x,y) = 1 - \frac{r_{x,y}}{p_x}$. The constraints in \eqref{eq:channel_lp} readily give $\sum_{x} \mu(x) = 1$ and $T(x,y) \in [0,1]$ for all $x,y$. We then have that for any distribution $\nu$ on $Y$,  
\begin{align}
\sum_{x,y} \mu(x) \nu(y) T(x,y) &\geq \min_{y} \sum_{x} \frac{p_x}{k} \cdot (1 - \frac{r_{x,y}}{p_x}) \\
 &= \min_{y} 1 - \frac{1}{k} \sum_{x} r_{x,y} \\
 &\geq 1 - \frac{1}{k} \ .
\end{align}
In addition, we have
\begin{align}
\sum_{x, y} \mu(x) W(y|x) T(x,y) &= 1 - \frac{1}{k} \sum_{x,y} r_{x,y} W(y|x) \\
&\geq 1 - \channel^{\NS}(W, k) \ .
\end{align}
As a result, there exists $\mu$ such that for all $\nu$, $\beta_{1-1/k}(\mu \cdot \nu, \mu \cdot W) \leq 1 - \channel^{\NS}(W,k)$.

For the other direction, we first show using linear programming duality that for any $\mu$,
\begin{align}
\label{eq:exchange_nu_T}
\max_{\nu \in \dist(Y)} \beta_{1-1/k}(\mu \cdot \nu, \mu \cdot W)
&= \min_{\substack{T : X \times Y \to [0,1] \\ \forall y, \; \sum_{x} \mu(x) T(x,y) \geq 1 - \frac{1}{k}}} \sum_{x,y} \mu(x) W(y|x) T(x,y) \ .
\end{align}
In order to see this, observe that $\beta_{\alpha}(\mu \cdot \nu, \mu \cdot W)$ is a linear program and thus using duality can also be written as a maximization program. In fact, 
\begin{align}
\beta_{\alpha}(\mu \cdot \nu, \mu \cdot W) &= \max_{\substack{\lambda_1 \geq 0, \lambda_2(x,y) \geq 0 \\ \mu(x)W(y|x) + \lambda_2(x,y) \geq \lambda_1 \mu(x) \nu(y)}} \lambda_1 \alpha - \sum_{x,y} \lambda_2(x,y)
\end{align}
As a result,
\begin{align}
\max_{\nu \in \dist(Y)} \beta_{\alpha}(\mu \cdot \nu, \mu \cdot W) &= \max_{\substack{\lambda_1 \geq 0, \lambda_2(x,y) \geq 0 \\ \mu(x)W(y|x) + \lambda_2(x,y) \geq \lambda_1 \mu(x) \nu(y) \\ \nu(y) \geq 0, \sum_y \nu(y) = 1}} \lambda_1 \alpha - \sum_{x,y} \lambda_2(x,y) \\
&= \max_{\substack{\lambda_1(y) \geq 0, \lambda_2(x,y) \geq 0 \\ \mu(x)W(y|x) + \lambda_2(x,y) \geq \lambda_1(y) \mu(x) }} \alpha \sum_{y} \lambda_1(y) - \sum_{x,y} \lambda_2(x,y) \ ,
\end{align}
where we simply set $\lambda_1(y) = \lambda_1 \nu(y)$. To conclude the proof of \eqref{eq:exchange_nu_T}, it suffices to observe that this last expression is nothing but the dual program for the right hand side of \eqref{eq:exchange_nu_T}.

Now given a distribution $\mu$ on $X$ and a test $T$ that satisfies the constraint on the right rand side of~\eqref{eq:exchange_nu_T}, we define $p_x = k \cdot \mu(x)$ and $r_{x,y} = k \cdot \mu(x) (1-T(x,y))$. Then the condition $\sum_{x} \mu(x) T(x,y) \geq 1 - \frac{1}{k}$ translates to $\sum_{x} \mu(x) - \frac{1}{k} \sum_{x} r_{x,y} \geq 1-\frac{1}{k}$. In other words, $\sum_{x} r_{x,y} \leq 1$. In addition, we clearly have $\sum_{x} p_{x} = k$ and $r_{x,y} \leq p_x$. This concludes the proof of the claim.
\end{proof}

  \bibliographystyle{abbrv}
  
  \bibliography{big}

\end{document}